\documentclass[11pt]{article}
\usepackage[utf8]{inputenc}
\usepackage[american]{babel}
\usepackage[T1]{fontenc}
\usepackage{amssymb}
\usepackage{amsmath}
\usepackage{amsthm}
\usepackage{graphicx}
\usepackage{cite}

\usepackage[plainpages=false]{hyperref}
\hypersetup{
	unicode=false,
	colorlinks=true,
	linkcolor=blue,
	citecolor=blue,
	filecolor=blue,
	urlcolor=blue,
	pdfauthor={Matthias Fischer, Daniel Jung, Friedhelm Meyer auf der Heide},
	pdfsubject={Gathering Anonymous, Oblivious Robots on a Grid},
	pdftitle={Gathering Anonymous, Oblivious Robots on a Grid},
	pdfkeywords={gathering problem, autonomous robots, distributed algorithms, local algorithms, mobile agents, runtime bound, swarm formation problems}
}

\newtheorem{theorem}{Theorem}
   \newtheorem{lemma}{Lemma}
   
\newtheorem{definition}{Definition}

\newcommand{\refsec}[1]{Section~\ref{#1}}
\newcommand{\refssec}[1]{Subsection~\ref{#1}}
\newcommand{\reffig}[1]{Figure~\ref{#1}}

\newcommand{\reflem}[1]{Lemma~\ref{#1}}

\newcommand{\refthm}[1]{Theorem~\ref{#1}}

\newcommand{\calO}{\mathcal{O}}
\newcommand{\calS}{\mathcal{S}}

\newcommand{\calX}{\mathcal{X}}
\newcommand{\calY}{\mathcal{Y}}
\newcommand{\calZ}{\mathcal{Z}}

\newcommand{\area}{\emph{Area}}

\newcommand{\boundary}{\emph{Boundary}}
\newcommand{\convex}{\emph{Convex}}

\newcommand{\viewingradius}{7}

\newcommand{\DiagA}{Diag-$\mathrm{A}$}
\newcommand{\DiagB}{Diag-$\mathrm{B}$}
\newcommand{\DiagAB}{Diag-$\{\mathrm{A}, \mathrm{B}\}$}

\title{Gathering Anonymous, Oblivious Robots on a Grid\\[1.5ex][Full Version]}
\author{
    Matthias Fischer
    \and
    Daniel Jung
    \and
    Friedhelm Meyer auf der Heide
}

\date{
    Heinz Nixdorf Institute \& Computer Science Department\\[0.2em]
    University of Paderborn (Germany)\\[0.2em]
    Fürstenallee 11, 33102 Paderborn\\[0.2em]
\quad\\[0.2em]
    \texttt{\{mafi,daniel.jung,fmadh\}@uni-paderborn.de}
}

\begin{document}
\bibliographystyle{alpha}
\maketitle
\thispagestyle{empty}
\begin{abstract}
	We consider a swarm of $n$ autonomous mobile robots, distributed on a 2-dimensional grid.
	A basic task for such a swarm is the gathering process:
	All robots have to gather at one (not predefined) place.
    A common local model for extremely simple robots is the following:
    The robots do not have a common compass, only have a constant viewing radius,
    are autonomous and indistinguishable, can move at most a constant distance in each step, cannot communicate, are oblivious and
    do not have flags or states.
    The only gathering algorithm under this robot model, with known runtime bounds, needs $\calO(n^2)$ rounds and works in the Euclidean plane.
    The underlying time model for the algorithm is the fully synchronous $\mathcal{FSYNC}$ model.
	On the other side, in the case of the 2-dimensional grid, the only known gathering algorithms for the same time and a similar local model
    additionally require a constant memory, states and ``flags'' to communicate these states to neighbors in viewing range.
    They gather in time $\calO(n)$.

    In this paper we contribute the (to the best of our knowledge) first gathering algorithm on the grid
	that works under the same simple local model as the above mentioned Euclidean plane strategy, i.e., without memory (oblivious),
    ``flags'' and states.
    We prove its correctness and an $O(n^2)$ time bound in the fully synchronous $\mathcal{FSYNC}$ time model.
	This time bound matches the time bound of the best known algorithm for the Euclidean plane mentioned above.
    We say gathering is done if all robots are located within a $2\times 2$ square, because in $\mathcal{FSYNC}$ such configurations
    cannot be solved.
\end{abstract}

\section{Introduction}
	Swarm robotics considers large swarms of relatively simple mobile robots deployed to some two- or three-dimensional area.
	These robots have very limited sensor capabilities; typically they can only observe aspects of their local environment.
	The objective of swarm robotics is to understand which global behavior of a swarm is induced by local strategies,
	simultaneously executed by the individual robots.
	Typically, the decisions of the individual robots are based on local information only.

	In order to formally argue about the impact of such local decisions of the robots on the overall behavior of the swarm,
	many simple models of robots, their local algorithms, the space they live in, and underlying time models are proposed.
	For a survey see the book \cite{flocchinioverview} by Flocchini, Prencipe, and Santoro.

	A basic desired global behavior of such a swarm is the gathering process: All robots have to gather at one (not predefined) place.
	Local algorithms for this process are defined and analysed for a variety of models
    \cite{gatheringthetanquadrat, MINCH, flocchinioverview, gathering-icalp, JungSPAA2016, JungIPDPS2016}.

    A common local model for extremely simple robots is the following:
    There is no global coordinate system.
    The robots do not have a common compass, only have a constant viewing radius,
    are autonomous and indistinguishable, can move at most a constant distance in each step, cannot communicate, are fully oblivious and
    do not have flags or lights to communicate a state to others.
    In this very restricted robot model, a robot's decision about its next action can only be based on the current relative positions of the
    otherwise indistinguishable other robots in its constant sized viewing range, and independent on past decisions or information (oblivious).

    The only gathering algorithm under this robot model, with known runtime bounds, needs $\calO(n^2)$ rounds and works in the Euclidean plane.
    The underlying time model for the algorithm is the fully synchronous $\mathcal{FSYNC}$ model (see \cite{localgathering, flocchinioverview}).
    In $\mathcal{FSYNC}$, all robots are always active and do everything synchronously.
    Time is subdivided into equally sized rounds of constant lengths.
    In every round all robots simultaneously execute their operations in the common \emph{look-compute-move model} \cite{Cohen:2004a} (\refsec{sec:OurLocalModel}).

    In the discretization of the Euclidean plane, the two-dimensional grid, under the same time and robot model,
    no runtime bounds for gathering are known.
    The concept of the Euclidean algorithm \cite{gatheringthetanquadrat} cannot be transferred to the grid, because it must be able
    to compute the center of the minimum enclosing circle of the robots in its viewing range (and then move to this position) and
    furthermore move arbitrary small distances.
    This clearly is impossible on the grid. Instead, completely different approaches are needed.
    
    In the only known gathering algorithms on the grid
    under the same time and a similar robot model, the robots need states (so-called runs) and flags
    to communicate these states to neighbors, and have to be able to memorize a fixed number of steps \cite{JungSPAA2016,JungIPDPS2016}.
    There, a robot with an active run state can further move this state to a neighboring robot.
    This allows coordinated robot operations over several consecutive rounds.
    In \cite{JungSPAA2016,JungIPDPS2016}, these operations are crucial for total running time proof ($\calO(n)$ rounds).

    In the current submission, \underline{we drop} the additional robot capabilities \underline{memory} and \underline{flags} or \underline{lights} to communicate a state to others.
    Then, analogously to the Euclidean strategy \cite{gatheringthetanquadrat}, explained above, a robot's decision about its next action
    can only be based on the current relative positions of the otherwise indistinguishable other robots in its constant sized viewing range,
    and independent on past decisions or information (oblivious).
    Especially coordinated robot operations over several consecutive rounds that are used in \cite{JungSPAA2016,JungIPDPS2016} cannot be performed under this more restricted model.
        
    To the best of our knowledge we
    present the first strategy \underline{under this restricted model on the grid} and prove a total running time of $\calO(n^2)$ rounds which complies with the best known running time for the Euclidean strategies in this model \cite{gatheringthetanquadrat}.
    More precisely, the running time of our strategy depends quadratically on the outer boundary length of the swarm.
	The outer boundary is the seamless sequence of neighboring robots that encloses all the others robots inside.

    We conjecture that $\Omega(n^2)$ is a lower bound for the number of rounds needed for our algorithm and, more generally,
    even for any algorithm within our restricted model.
    At least for our algorithm, we conjecture that a worst case instance is a configuration with robots on the boundary of an
    axis-parallel square. Experiments support this conjecture.
    Our paper has been accepted for ALGOSENSORS 2017 \cite{JungALGO2017}.
\section{Related Work}\label{sec:relatedwork}
    There is vast literature on robot problems researching how specific coordination problems can be solved by a swarm of robots given a certain limited set of abilities.
    The robots are usually point-shaped (hence collisions are neglected) and positioned in the Euclidean plane.
    They can be equipped with a memory or are \emph{oblivious}, i.e., the robots do not remember anything from the past and perform their actions only on their current views.
    If robots are anonymous, they do not carry any IDs and cannot be distinguished by their neighbors.
    Another type of constraint is the compass model:
    If all robots have the same coordinate system, some tasks are easier to solve than if all robots' coordinate systems are distorted.
    In \cite{gathering-compasses,Izumi2012} a classification of these two and also of dynamically changing compass models is considered, as well as their effects regarding the gathering problem in the Euclidean plane.
    The operation of a robot is considered in the \emph{look-compute-move model} \cite{Cohen:2004a}.
    How the steps of several robots are aligned is given by the \emph{time model}, which can range from an asynchronous $\mathcal{ASYNC}$ model (e.g., see \cite{Cohen:2004a}), where even the single steps of the robots' steps may be interleaved, to a fully synchronous $\mathcal{FSYNC}$ model (e.g., see \cite{localgathering}), where all steps are performed simultaneously.
    A collection of recent algorithmic results concerning distributed solving of basic problems like gathering and pattern formation, using robots with very limited capabilities, can be found in the book \cite{flocchinioverview} by Flocchini et al.
   
    One of the most natural problems is to gather a swarm of robots in a single point.
    Usually, the swarm consists of point-shaped, oblivious, and anonymous robots. The problem is widely studied in the Euclidean plane.
    Having point-shaped robots, collisions are understood as merges/fusions of robots and interpreted as gathering progress \cite{MINCH,gatheringthetanquadrat, JungSPAA2016}.
    In \cite{gathering-icalp} the first gathering algorithm for the $\mathcal{ASYNC}$ time model with multiplicity detection (i.e., when a robot can detect if other robots are also located at its own position) and global views is provided.
    Gathering in the local setting was studied in \cite{localgathering}.
    In \cite{impossibilityofgathering} situations when no gathering is possible are studied.
    The question of gathering on graphs instead of gathering in the plane was considered in \cite{practicalrendevouzaktuell, rendezvousingraphen, gatheringOnRing}.
    In \cite{Stefano2013} the authors assume global vision, the $\mathcal{ASYNC}$ time model and furthermore allow unbounded (finite) movements.
    They show optimal bounds concerning the number of robot movements for special graph topologies such as trees and rings.
    
    Concerning the gathering on grids,
    in \cite{gatheringongrids} it is shown that multiplicity detection is not needed and the authors further provide a characterization of solvable gathering configurations on finite grids.
    In \cite{OptExactGatheringGrids2014}, these results are extended to infinite grids, assuming global vision.
    The authors characterize \emph{gatherable} grid configurations concerning exact gathering in a single point.
    Under their robot model and the $\mathcal{ASYNC}$ time model, the authors present an algorithm which gathers \emph{gatherable} configurations
    optimally concerning the total number of movements.
    
    Assuming only local capabilities of the robots, esp.\ only local vision and no compass, makes gathering challenging.
    For example, a given global vision, the robots could compute the center of the globally smallest enclosing square or circle and just move to this point.
    For gathering with presence of a global compass, the authors in \cite{Cellular2016} provide a simple gathering algorithm:
    The robots from the left and right swarm boundaries keep moving towards the swarm's inside.
    In some kind of degenerated cases, instead the robots on the top and bottom boundaries do this.
    
    In the $\mathcal{FSYNC}$ time model, the total running time is a quality measure of an algorithm.
    In this time model there exist several results that prove runtime bounds \cite{gatheringthetanquadrat,hopper,gtm,JungIPDPS2016,JungSPAA2016}.
    For local robot models, the locality strongly restricts the robot capabilities: no global control, no unique IDs, no compass, only local vision (i.e., they can only see other robots up to a constant distance) and no (global) communication.
    But even under this strongly local model, the presence of remaining local capabilities such as allowing a constant number of states,
    constant memory or locally visible states (flags, lights), can drastically change running times by even more than the factor $n$ \cite{hopper,JungIPDPS2016,JungSPAA2016}.
    The price for this improvement then is many more complicated strategies.
    
    One example are strategies that maintain and shorten a communication chain between an explorer and a base camp.
    The \emph{Hopper} and \emph{Manhattan Hopper} strategies \cite{hopper} solve this problem in time $\calO(n)$,
    in the Euclidean plane and on the grid, respectively, using robots with a constant number of states, a constant
    memory and the capability to communicate states to local neighbors (flags, lights).
    Without these additional robot capabilities, the simple Euclidean \emph{Go-To-The-Middle} strategy \cite{gtm} needs notably more time $\calO(n^2\log(n))$ for solving the same problem.
    Concerning the gathering under this restricted model, the simple, Euclidean \emph{Go-To-The-Center} strategy
    \cite{gatheringthetanquadrat} needs time $\calO(n^2)$.
    (A faster strategy for the Euclidean plane does not exist, yet, and it is still unknown if this bound is tight.)
    On the grid, two asymptotically optimal $\calO(n)$ strategies exist that, solve the gathering of an arbitrary
    connected swarm \cite{JungSPAA2016} and the gathering of a closed chain of robots \cite{JungIPDPS2016}, respectively.
    Like the above communication chain strategies, they require more complex robots with a const.\ number of states,
    a const.\ memory and the capability to communicate states to local neighbors (flags, lights).
    Strategies without these additional capabilities do not exist, yet.

    In the present paper, we deliver such an algorithm that uses the same strongly restricted model as the Euclidean \emph{Go-To-The-Center} gathering strategy \cite{gatheringthetanquadrat}.
    Our strategy gathers in time $\calO(|\mathrm{outer\ boundary}|^2)\subseteq\calO(n^2)$, where
    $|\mathrm{outer\ boundary}|$ denotes the length of the swarm's outer boundary (cf.\ \reffig{fig:boundaryfullex_convexvertexdefi}.$(i)$) which naturally
    is $\in\calO(n)$.
    This is comparable to the $\calO(n^2)$ bound for the Euclidean gathering \cite{gatheringthetanquadrat}.
    We conjecture that $\calO(|\mathrm{outer\ boundary}|^2)$ is also tight on the grid.
\section{Our Local Model}\label{sec:OurLocalModel}
        Our mobile robots need very few and simple capabilities:
        A robot moves on a two-dimensional grid and can change its position to one of its eight horizontal, vertical or diagonal neighboring grid cells.
		It can see other robots only within a constant \emph{viewing radius} of \viewingradius\ (measured in $L_1$-distance).
		We call the range of visible robots the \emph{viewing range}.
        Within this viewing range, a robot can only see the relative positions of the viewable robots.
        The robots have no compass, no global control, and no IDs.
        They cannot communicate,
        do not have any states (no flags, lights) and are oblivious.

        Our algorithm uses the fully synchronous time model $\mathcal{FSYNC}$,
		in that all robots are always active and do everything synchronously.
        Time is subdivided into equally sized rounds of constant lengths.
        In every round all robots simultaneously execute their operations in the common \emph{look-compute-move model} \cite{Cohen:2004a}, which divides one operation into three steps.
        Every round contains only one cycle of these steps:
        In the \emph{look} step, the robot gets a snapshot of the current scenario from its own perspective, restricted to its constant-sized viewing range.
        During the \emph{compute} step, the robot computes its action, and eventually performs it in the \emph{move} step.
        If a robot has moved to an occupied grid cell, the robots from then on behave like one robot.
        We say they \emph{merge} and remove one of them.

		We say gathering is done if all robots are located within a $2\times 2$ square, because such configurations cannot
		be solved in our time model.

        The swarm must be connected.
        In our model, two robots are connected if they are located in horizontal or vertical neighboring grid cells.
        The operations of our algorithm do not destroy this connectivity.
\section{The Algorithm}
    \label{sec:algo}
    A robot decides to hop on one of its $8$ neighboring grid cells only dependent on the current robot positions within its viewing range.
    We distinguish diagonal (\refssec{ssec:diaghops}) and horizontal/vertical (\refssec{ssec:horizverthops}) hops.
    The hops are intended to achieve the gathering progress by modifying the swarm's \emph{outer boundary}.
    \reffig{fig:boundaryfullex_convexvertexdefi}.$(i)$ defines the swarm's boundaries:
    Black and hatched robots are \emph{boundary} robots.
    The boundary on which the black robots are located borders the swarm and is called the swarm's \emph{outer boundary}.
    In the figure, all other robots are colored grey. White cells are empty.
    \subsection{Diagonal hops}\label{ssec:diaghops}
        \begin{figure}[h]
        \centering
            \includegraphics[width=0.85\textwidth]{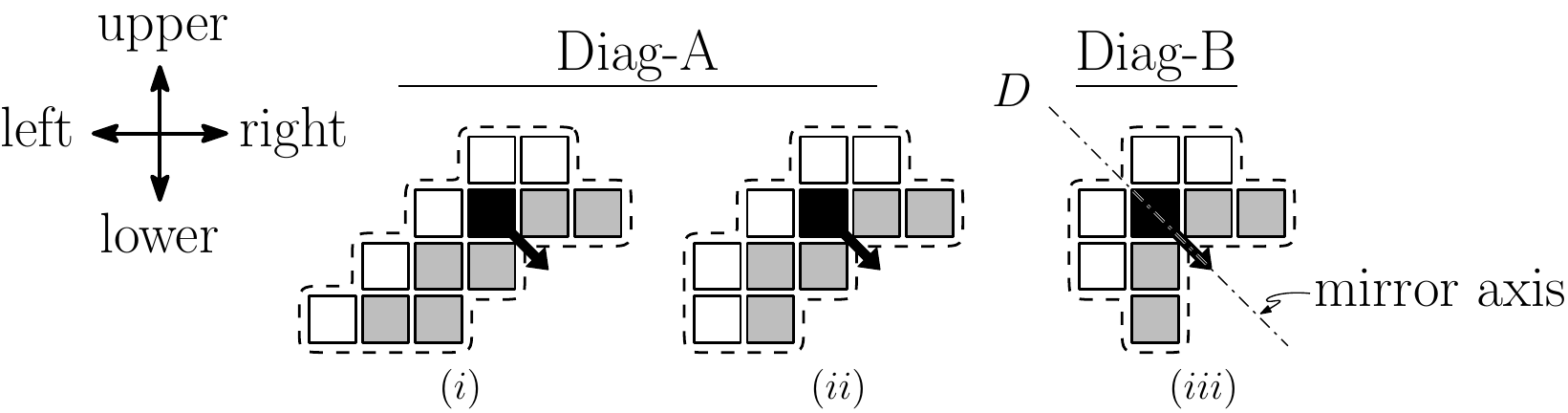}
            \caption{\textbf{Hop patterns:} One of the \DiagA\ or \DiagB\ Hop patterns must match the relative robot positions within the black robot's viewing range.
                This is the hop criterion, necessary for allowing the black robot to perform the depicted diagonal hop.
                In this paper, the notation \DiagAB\ means ``\DiagA\ or \DiagB''.
                }
            \label{fig:hops_diag}
        \end{figure}
        \begin{figure}[b]
        \centering
            \includegraphics[width=0.85\textwidth]{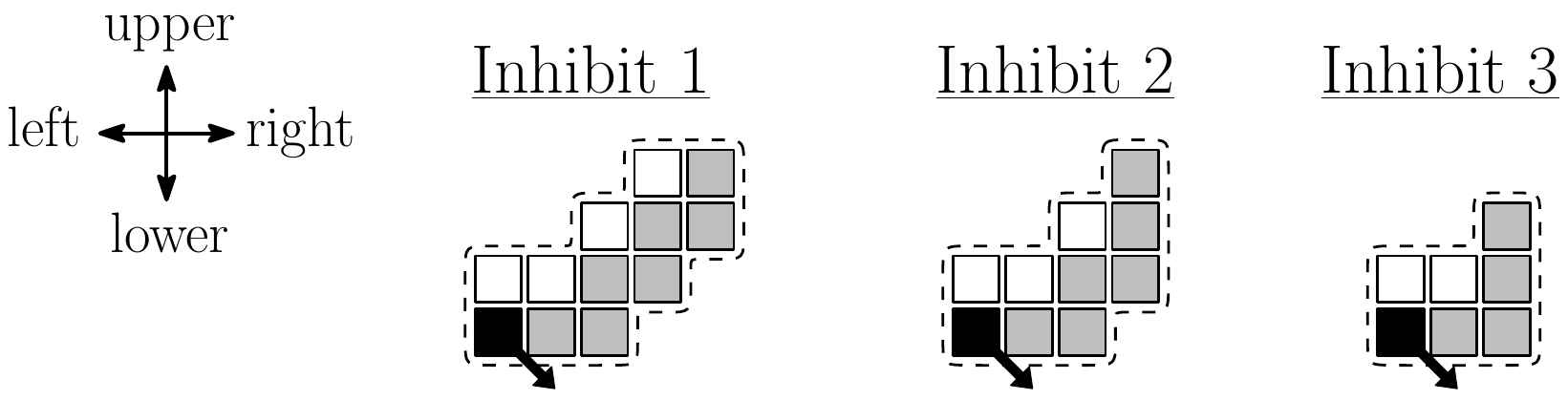}
            \caption{\textbf{Inhibit patterns:} Patterns that, in case they match, inhibit the black robot's hop.}
            \label{fig:hops_diag_inhibit}
        \end{figure}
        If a robot $r$ (marked black in Fig. \ref{fig:hops_diag}/\ref{fig:hops_diag_inhibit}) checks whether it can execute a diagonal hop, it compares the patterns of \reffig{fig:hops_diag} and \ref{fig:hops_diag_inhibit}
        to the robot positions in its viewing range:
        Robot $r$ checks if one of the \DiagAB\ Hop patterns matches the current scenario from its own perspective.
        Patterns that are created by an arbitrary horizontal and vertical mirroring and an arbitrary 90 degree rotation of the three patterns of \reffig{fig:hops_diag} are also valid and have to be checked.

		Depending on the matching Hop pattern the robot does the following: \label{enum:diagab}
        \begin{enumerate}
            \item If a \DiagA\ pattern matches, then robot $r$ checks if, using the same rotation and mirroring, any of the Inhibit patterns
                match the \textit{upper-right area} of its viewing range.
                If at least one Inhibit pattern matches, then the \DiagA\ hop of robot $r$ is not executed.
                Otherwise, if none of the Inhibit patterns match, the robot $r$ hops according to the matching \DiagA\ pattern.\label{enum:DiagAInhibitCheck}
            \item If the \DiagB\ pattern matches, then the robot $r$ checks if, using the same rotation and mirroring, also any of the
           		Inhibit patterns match the \textit{upper-right area} and the \emph{lower-left area} of its viewing range.
               	However, in case of the \emph{lower-left area}, the Inhibit pattern has to be mirrored at the diagonal mirroring axis $D$ shown in \reffig{fig:hops_diag}.$(iii)$.
               	If for both areas matching Inhibit patterns have been found, then the \DiagB\ hop of robot $r$ is not executed.
                Otherwise the robot $r$ hops according to the matching \DiagB\ pattern.\label{enum:DiagBInhibitCheck}
        \end{enumerate}
    \subsection{Horizontal and vertical hops (HV hops)}\label{ssec:horizverthops}
        Robots can also hop in vertical or horizontal direction (HV hop).
        We allow these hops for length 1 and 2 (cf.\ Figure~\ref{fig:hops_merge12}).
        For length 2, horizontal or vertical hops, respectively, are a joint operation of two neighbouring robots.
        If for a robot a horizontal and a vertical HV hop apply at the same time (see $b^ \star$ of Fig. \ref{fig:hops_merge12}), then it instead
        performs a diagonal hop as shown in the figure.
        \begin{figure}[h]
            \centering
            \includegraphics[width=0.98\textwidth]{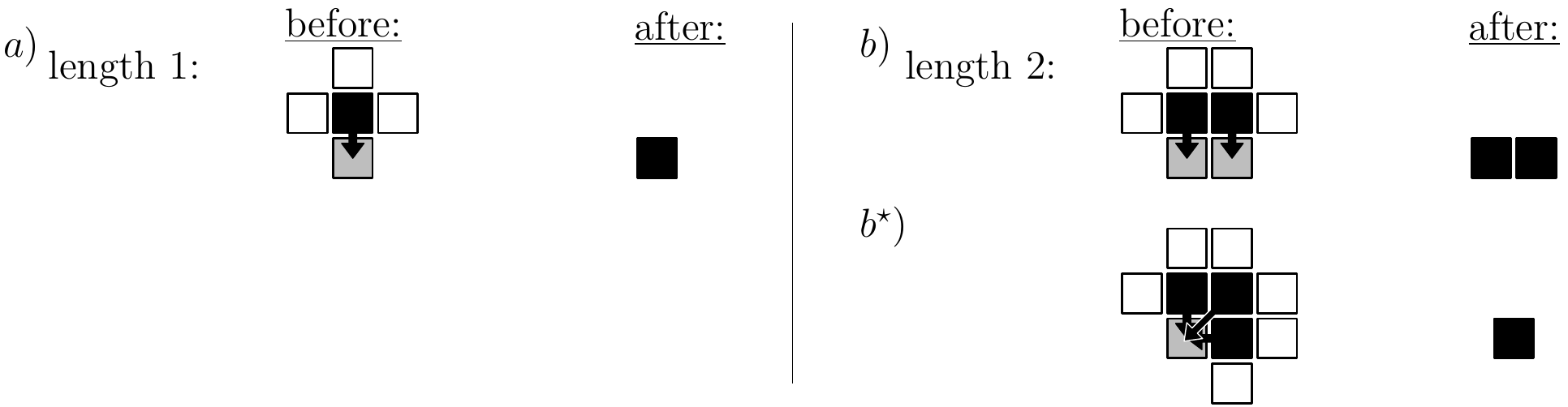}
            \caption{The black robots simultaneously hop downwards. Afterwards, robots that are located at the same position \emph{merge}.}
            \label{fig:hops_merge12}
        \end{figure}
        After a HV hop, every target cell contains at least two robots.
        We let these robots \emph{merge}: i.e., we remove all but one of the robots at the according cell.

        \DiagAB\ and HV hops are executed simultaneously in the same step of the algorithm.
        %
        %
        In summary, all robots synchronously execute the algorithm, shown in \reffig{fig:algo}.
        \begin{figure}[h]
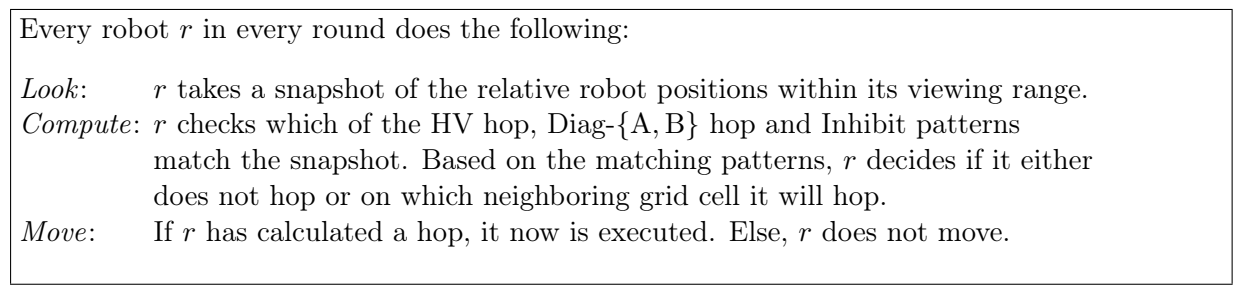

            \fbox{\parbox{0.97\columnwidth}{
                Every robot $r$ in every round does the following:
                \begin{tabbing}
                    \emph{Compute}: \= \kill
                    \emph{Look}:\> $r$ takes a snapshot of the relative robot positions within its viewing range.\\
                    \emph{Compute}:\> $r$ checks which of the HV hop, \DiagAB\ hop and Inhibit patterns\\\> match the snapshot. 
                        Based on the matching patterns, $r$ decides if it either\\\> does not hop or on which neighboring grid cell it will hop.\\
                    \emph{Move}:\> If $r$ has calculated a hop, it now is executed. Else, $r$ does not move.
                \end{tabbing}
            }}
            \caption{The algorithm.}
            \label{fig:algo}
        \end{figure}
\section{Measuring the Gathering Progress}
    The gathering progress measures that we will use for the analysis of our strategy are heavily dependent on the length and shape
    of the swarm's outer boundary.
	In order to analyze these measures, we need the terms boundary, outer boundary, length, as well as convex and concave vertices.

 	\paragraph{{\bfseries Swarm's boundary.}}
	The swarm's boundary is the set of all robots that have at least one empty adjacent cell in a horizontal, vertical, or diagonal direction.
	\reffig{fig:boundaryfullex_convexvertexdefi}.$(i)$ shows an example: 
	Black and hatched robots are boundary robots.
	The empty cells contain no robot and are colored in white.
	When speaking about a \emph{subboundary}, we mean a connected sequence of robots of some boundary.
 	\paragraph{{\bfseries Swarm's outer boundary.}}    
    The swarm's outer boundary is the boundary that borders on the outside of the swarm.
    In \reffig{fig:boundaryfullex_convexvertexdefi}.$(i)$ the black robots belong to outer boundary.
    All other robots are not part of the boundaries, i.e., they have an adjacent robot in all directions (horizontal, vertical, and diagonal).
    In \reffig{fig:boundaryfullex_convexvertexdefi}.$(i)$ they are colored grey.
    \begin{figure}[h]
       \begin{center}
           \includegraphics[width=\textwidth]{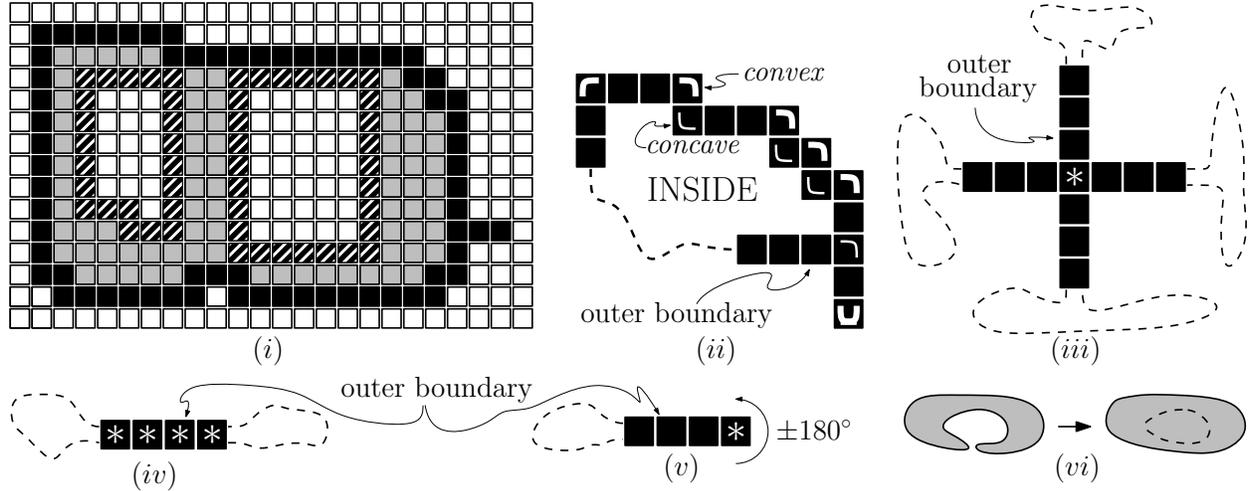}
       \end{center}
       \caption{$(i)$: Definition of the (\emph{outer}) \emph{boundary}. Outer boundary: black robots; $(ii)$: Definition of \emph{convex} and \emph{concave} \emph{vertex}. Convex vertices: fat curves; $(iii)$: Cross shape. The ``$\ast$'' marked robot is counted four times; $(iv)$: Hourglass shape: The ``$\ast$'' marked robots are counted twice; $(v)$: $\pm 180^\circ$ rotation. The ``$\ast$'' marked robot is counted twice.; $(vi)$: During the gathering, inner bubbles can be developed.}
       \label{fig:boundaryfullex_convexvertexdefi}
    \end{figure}
    \paragraph{{\bfseries Outer boundary's length.}}
    We measure the outer boundary's length as follows:
    We start at a cell of the outer boundary and perform a complete walk along this boundary while we define the
    \emph{length} as the total number of steps that we performed during this walk.
    This means that if the swarm is hourglass- or cross-shaped, for example, some robots are counted multiple (up to four) times (cf.\ \reffig{fig:boundaryfullex_convexvertexdefi}.$(iii,iv)$).
    Furthermore, robots on which a turn by $\pm 180^\circ$ is performed during the walk are counted twice (cf.\ \reffig{fig:boundaryfullex_convexvertexdefi}.$(v)$).
    We denote this length by $|\mathrm{outer\ boundary}|$.
	\paragraph{{\bfseries Convex and concave vertices.}}
    On the boundary we further distinguish \emph{convex} and \emph{concave vertices}.
    A vertex of the boundary is a robot that looks like a corner.
    In \reffig{fig:boundaryfullex_convexvertexdefi}.$(ii)$ fat curves mark \emph{convex vertices} of the swarm's outer boundary, while
    the thin curves mark the \emph{concave vertices} of the swarm's outer boundary.
    \paragraph{{\bfseries Outline of the running time proof  --- How we get gathering progress.}}
        We distinguish three kinds of progress measures that help us to prove the quadratic running time.
        \begin{tabbing}
            \textbf{\boundary:} \=\kill
            \textbf{\boundary:} \> Length of the swarm's outer boundary.\label{enum:lengthboundary_intro}\\
            \textbf{\convex:} \> Difference between the number of convex vertices on the swarm's\\\> outer boundary and its maximum value.\label{enum:numberconvex_intro}\\
            \textbf{\area:} \> Included area.\label{enum:includedarea_intro}
        \end{tabbing}
        We have designed the hops in such a way that the length of the outer boundary (\boundary) never increases.
        But it can remain unchanged over several rounds.
        Then, instead, we measure the progress by \convex.
        As we draw robots as squares, the total number of convex vertices on the swarms' outer boundary is naturally upper bounded by its maximum value $4|\mathrm{outer\ boundary}|$.
        \convex\ is the difference between this maximum value and the actual number of convex vertices on the outer boundary.
        We will show that also \convex\ never increases.

        In rounds in which both \boundary\ and \convex\ do not achieve progress, we instead measure the gathering progress by \area.
        We measure \area\ as the number of robots on the swarm's outer boundary plus the number of inside cells (occupied as well as empty ones).
        In contrast to the other progress measures, \area\ does not decrease monotonically in general, but we show that it decreases monotonically in rounds
        without \boundary\ and \convex\ progress.
        We upper bound the size by that the \area\ can instead be increased during other rounds and show that this
        makes the total running time worse at most by a constant factor.

        All three measures depend only on the length of the swarm's outer boundary.
        While \boundary\ and \convex\ are linear, \area\ is quadratic.

        This then leads us to a total running time $\calO(|\mathrm{outer\ boundary}|^2)$.

\section{Correctness \& Running Time}\label{sec:corrruntime}
    In this section, we formally prove the correctness of the progress measures and
    finally the total running time (\refthm{thm:runningtime}). 
    \subsection{Progress Measure \boundary} 
        \begin{lemma}\label{lem:boundaryprogress}
            During the whole gathering, \boundary\ is monotonically decreasing.
        \end{lemma}
        \begin{proof}
            As the definition of HV hops requires that the robots hop onto occupied cells, such hops naturally cannot
            increase the number of robots on the outer boundary.
            So we consider \DiagAB\ hops in which robots hop towards the swarm's outside.
            In order to increase the number of robots on the outer boundary, the target cell of such hops
            must be empty.
            But then, the hopped
            robot has also been part of the outer boundary before the hop, so that the boundary length did not increase.
        \end{proof}
    \subsection{Impact of Inhibit Patterns: \emph{Collisions}}\label{ssec:inhibitimpact}
        For the proofs of the progress measures \convex\ and \area, we need a deeper insight why certain robot hops are inhibited
        by Inhibit patterns (cf.\ \reffig{fig:hops_diag_inhibit}).
        When proving that the \convex\ progress is monotonically decreasing (\reflem{lem:convexprogress}), we analyze the 
        change of the total number of convex vertices that is induced by the robot hops.
        This is only possible if certain simultaneously hopping robots are not too close together.
        Inhibit patterns ensure this minimum distance.

        Cf.\ \reffig{fig:hops_diag} and \ref{fig:hops_diag_inhibit}.
        From a more global point of view, the Inhibit patterns ensure that the black robot only performs its \DiagA\ or \DiagB\ hop,
        respectively, if the next robot(s) at distance $2$ along the boundary does (do) not perform a hop in the opposite direction.
        If the hop of the black robot is blocked by Inhibit patterns, we will say the robots \emph{collide}.
        This will be used in the proofs of our progress measures.
        \reffig{fig:collision_ex} shows significant collision examples:
        $(i)$: For both, $r$ and $r'$ \DiagA\ matches.
        Without inhibition patterns, $r$ would hop to the lower right, while $r'$ would hop in the opposite direction to the upper left.
        But the Inhibit 1 pattern inhibits the hop of $r$:
        $r$ collides with $r'$.
        Analogously also the hop of $r'$ can be inhibited.
        $(ii)$: If \DiagB\ matches for $r$, then the hop is inhibited if concerning both $r'$ and $r''$ an inhibition pattern matches.
        In this example, for $r'$ a \DiagA\ pattern and for $r''$ a \DiagB\ pattern could else enable hops in the opposite direction than the hop of $r$,
        but the matching Inhibit 1 and 2 patterns inhibit the hop of $r$: $r$ collides with $r',r''$.
        A more detailed analysis of collisions is provided in the proofs for \reflem{cor:area}.
        \begin{figure}[h]
        \centering
            \includegraphics[width=0.9\textwidth]{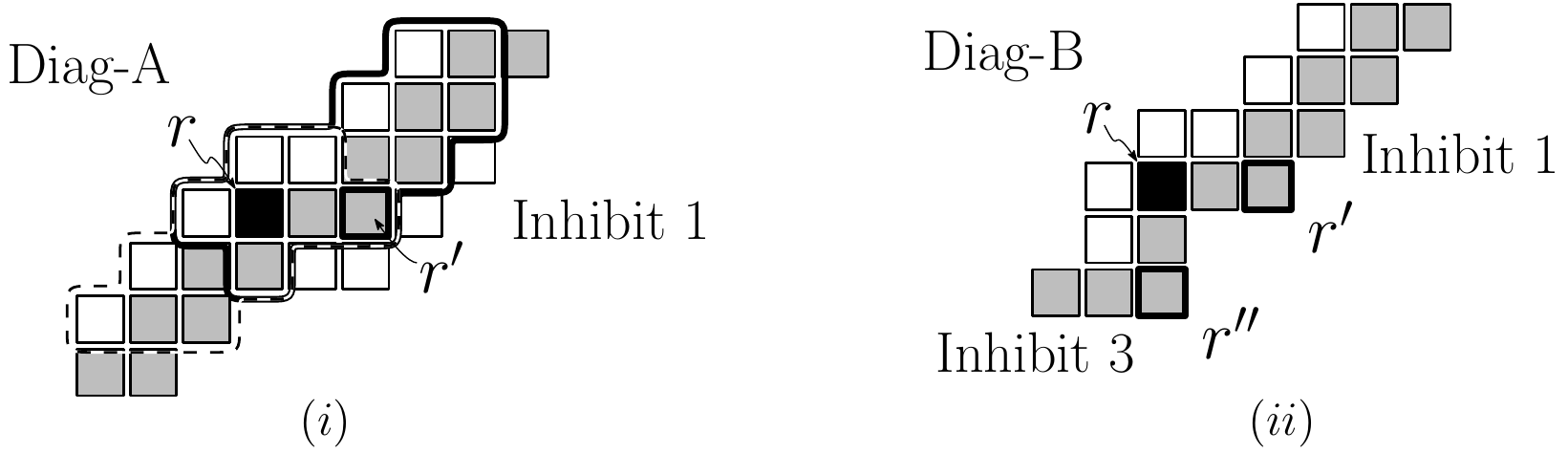}
            \caption{\emph{Collisions}. $(i)$: For $r$, \DiagA\ and Inhibit 1 matches. $r$ does not hop. $(ii)$: For $r$, \DiagB\ and Inhibit 1 and 3 match. $r$ does not hop.}
            \label{fig:collision_ex}
        \end{figure}
    \subsection{Progress measure \convex}
        \begin{lemma}\label{lem:convexprogress}
            During the whole gathering, \convex\ is monotonically decreasing.
        \end{lemma}
        \begin{proof}
            If we say ``convex/concave vertices'', we consider only the outer boundary.
            First, we analyze the HV hops.
            Here, a HV hop can reduce the number of convex vertices by at most 2.
            At the same time, the outer boundary becomes shorter by at least 2.
            Then, \convex\ either remains unchanged or decreases.

            Concerning the \DiagAB\ hops, we look at \reffig{fig:convexanalysis}:
            \begin{figure}[h]
               \begin{center}
                   \includegraphics[width=\textwidth]{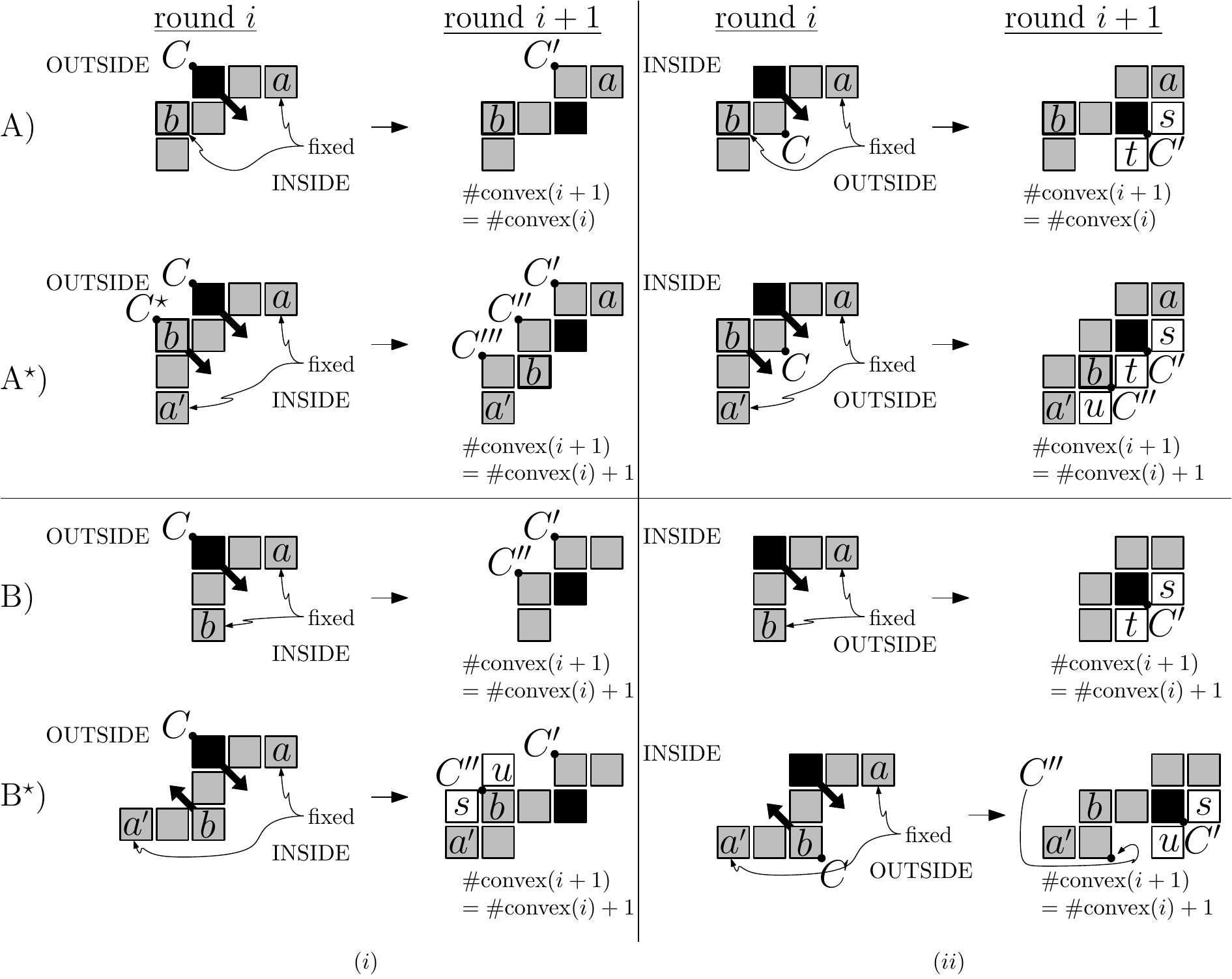}
               \end{center}
               \caption{Local effect of all kinds of hops on the number of convex vertices on the outer boundary. $C,C',C''$ denote the counted convex vertices. \#convex$(i)$ denotes the number of convex vertices in round $i$.}
               \label{fig:convexanalysis}
            \end{figure}
            The figure shows how the diagonal hops can (locally) change the number of convex vertices on the outer boundary.
            (In the figure, $(ii)$ shows the hops from $(i)$, but for switched INSIDE and OUTSIDE.)
            In all cases, the Inhibit patterns ensure that the robots $a,a'$ do not move (cf.\ \refssec{ssec:inhibitimpact}).
            In the figure, we distinguish \DiagA\ and \DiagB\ hops, while $\mathrm{A},\mathrm{A}^\star$ refer to \DiagA\ and $\mathrm{B},\mathrm{B}^\star$ to \DiagB.
            We distinguish the case that the robot $b$ does not hop ($\mathrm{A},\mathrm{B}$) and the other case, that it performs a hop ($\mathrm{A}^\star, \mathrm{B}^\star$).
            The result of the case distinction is that in column $(i)$ the number of convex vertices never decreases.
            In column $(ii)$, this is also the case if the white marked cells $s,t,u$ are empty.

            If instead not all of $s,t,u$ are empty, the number of convex vertices might also become smaller.
            But even in this case, still \convex\ progress does not increase:
            We now show that then also $|$outer boundary$|$ becomes smaller as well as the
            maximum value for the total number of convex vertices, so that by definition \convex\ progress
            is not increased, i.e., it still behaves monotonically:
            \begin{figure}[h]
               \begin{center}
                   \includegraphics[width=\textwidth]{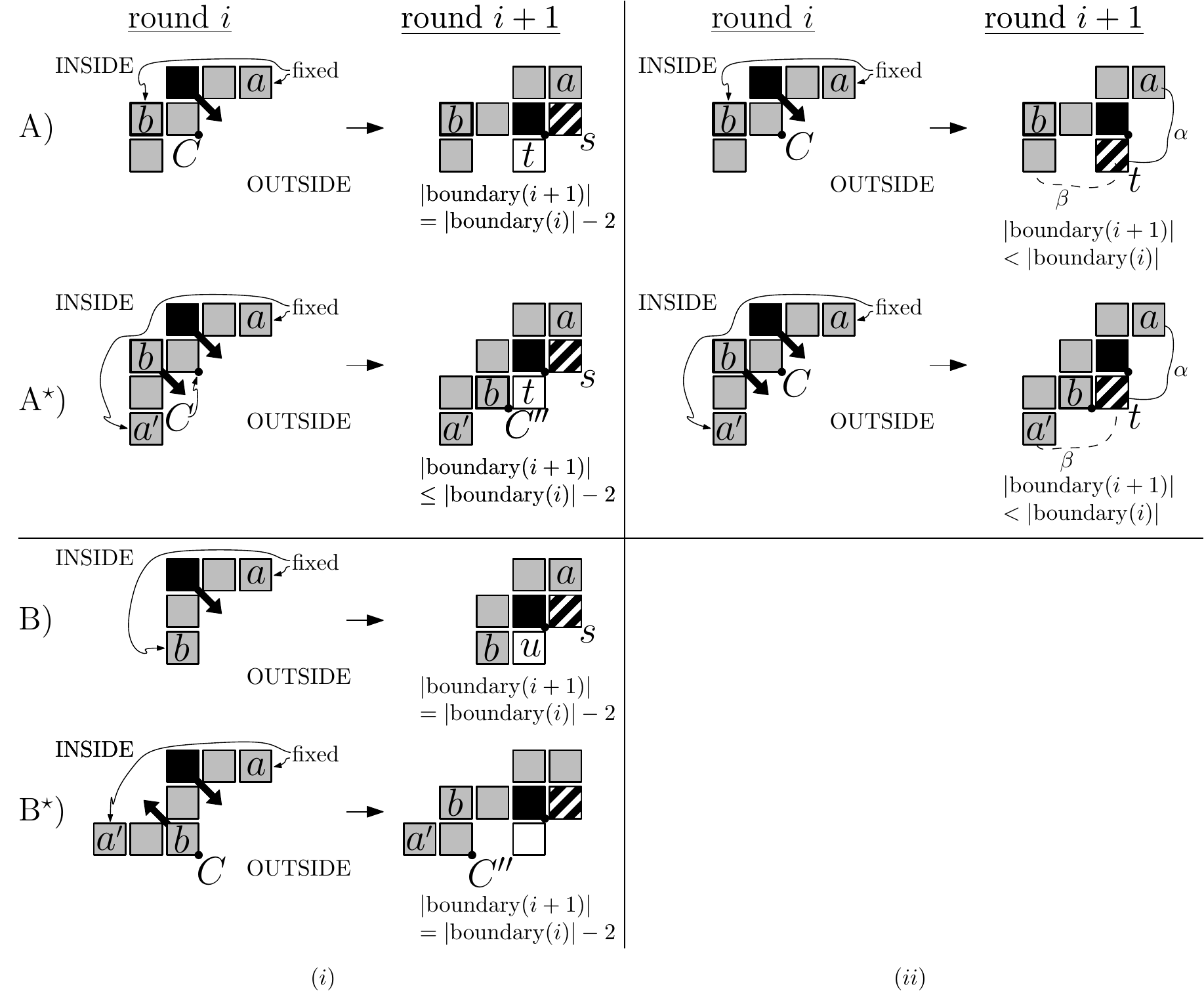}
               \end{center}
               \caption{Diagonal hops can change the outer boundary's length. $C,C',C''$ denote relevant convex vertices. \#boundary$(i)$ denotes the outer boundary length in round $i$.} 
               \label{fig:convexanalysis_boundary}
            \end{figure}
            Figure~\ref{fig:convexanalysis_boundary} shows the relevant cases.
            With reference to \reffig{fig:convexanalysis}.$(ii)$, all hops are performed towards the swarms' outside.
            In column $(i)$, only the cell $s$ contains a robot.
            There we see, that in all cases the outer boundary becomes shorter by at least 2.
            In case of $\mathrm{A}^\star$, this can be even more than $2$ for the case that
            after the hop the cell below $b$ (in \reffig{fig:convexanalysis}.$(ii)$, this was the cell $u$) also contains a robot.
            Column $(ii)$ shows the cases where cell $t$ is (also) occupied.
            Because the swarm is always connected, the hatched robot must be connected to the rest of the swarm.
            This can be either via the subboundary $\alpha$ or $\beta$.
            The hop shortens the outer boundary by forming inner bubbles (\reffig{fig:boundaryfullex_convexvertexdefi}.$(vi)$).
            %
        \end{proof}
    \subsection{Progress measure \area}
        The third progress measure \area\ does not behave monotonically.
        It can be increased during rounds where we get \boundary\ or \convex\ progress.
        But we use it for estimating the number of the remaining rounds (\reflem{cor:area}).
        And in the proof of \refthm{thm:runningtime} we show that the increased amount of the \area\ progress measure
        does not worsen the asymptotic running time.
        %
        \begin{lemma}\label{cor:area}
            If in a step of the gathering process neither \boundary\ nor \convex\ has progress, then instead \area\ has progress by at least $-8$.
        \end{lemma}
        For the proof of \reflem{cor:area}, see \refsec{sec:areaproofs}.
    \subsection{Total running time}
        Now we can combine all three progress measures \boundary, \convex\ and \area\ for the
        running time proof (\refthm{thm:runningtime}).
        \begin{theorem}\label{thm:runningtime}
            A connected swarm of $n$ robots on a grid
            can be gathered in
            \\
            $\calO(|\mathrm{outer\ boundary}|^2)\subseteq\calO(n^2)$ many rounds.
        \end{theorem}
        \begin{proof}
            Let $B$ be the initial length of the swarm's outer boundary.
            We know from \reflem{lem:boundaryprogress} that \boundary\ decreases monotonously.
            Then, progress in \boundary\ happens at most $B$ times.
            By \reflem{lem:boundaryprogress}, \convex\ also decreases monotonously.
            As every robot on the swarm's outer boundary can provide at most 4 convex vertices,
            \convex\ progress happens at most $4B$ times.

            We estimate the rounds without \boundary\ and \convex\ progress via the size of the included area,
            i.e., the \area\ progress.
            By \reflem{cor:area}, we know that in every round without \boundary\ and \convex\ progress, the area becomes smaller by at least $-8$.
            But, \area\ is not a monotone progress measure,
            in rounds with \boundary\ or \convex\ progress, the included area can increase:
            While HV hops cannot increase the included area, \DiagAB\ hops can.
            First, we assume that the according \DiagAB\ hops do not change the outer boundary length.
            Then, every time \convex\ has progress, the area can become larger by at most $B$, because
            every robot hop on the outer boundary can increase the area by at most 1.
            As \convex\ happens at most $4B$ times, this in total is upper bounded by $4B^2$.

            If the outer boundary length changes, i.e., becomes shorter, then the included area can increase (cf.\ proof of \reflem{lem:convexprogress} and \reffig{fig:boundaryfullex_convexvertexdefi}.$(vi)$).
            Then, a \boundary\ progress by $\ell$ can also increase the included area by $\Delta A\leq\ell^2\leq\ell B$.
            But as \boundary\ progress is monotonically decreasing, the sum of all these $\Delta A$ is upper bounded by $B^2$.

            Summing it up, during the whole process of the gathering, the area can be increased by at most $(4+1)B^2=5B^2$.
            Together with the initial area of at most $B^2$, \area\ progress happens at most $6B^2$.
            Then, the gathering is done after at most $B+4B+6B^2$ rounds.
        \end{proof}
\section{Proof of \reflem{cor:area}}\label{sec:areaproofs}
    In this section, we provide the proof of \reflem{cor:area}.
        In the following, we will prove \reflem{lem:area_nobridges} and \reflem{lem:area_bridges} whose combination immediately deliver the proof of \reflem{cor:area}.
        \reflem{lem:area_nobridges} and \ref{lem:area_bridges} are given later, when we have prepared for their proofs.
        \subsection{Outline of the proof}
            In \reflem{lem:area_nobridges}, we assume that the swarm does not contain hourglass shaped parts, while
            in \reflem{lem:area_bridges} the result is generalized to swarms that do contain  hourglass shaped parts.
            In \reffig{fig:treetransformation}, hourglass shaped parts are marked by black robots.
            We call them \emph{bridges}.
            A more detailed definition is given before \reflem{lem:area_nobridges}.

            For the proofs, we will first construct a more abstract model of the swarm's outer boundary:
            During a walk along the swarm's outer boundary one performs rotations by $0^\circ$ or $\pm 90^\circ$.
            We subdivide the outer boundary into parts (so-called \emph{supercorners}) such that at both of their endpoints
            a \DiagAB\ hop can be performed and that during a walk along such a supercorner, a total rotation by $0^\circ$ or $\pm 90^\circ$
            is performed.
            We will argue that because always a total rotation by $+360^\circ$ is performed during a complete counter-clockwise
            walk along the swarm's outer boundary, there must be 4 more $+90^\circ$ than $-90^\circ$ supercorners.
            In our construction, robot hops of $+90^\circ$ supercorners are always performed towards the swarm's inside.
            This then proves the statement of \reflem{lem:area_nobridges} that \area\ decreases by at least $-8$.

            When generalizing this to \reflem{lem:area_bridges}, we first notice that bridges can hinder hops towards the swarm's inside.
            But we will argue, that a single bridge cannot hinter the total $-8$ \area\ decrease that is proven by \reflem{lem:area_nobridges},
            and that at least two bridges are needed for this.
            We subdivide the swarm into bridges and subswarms that are connected by bridges, as shown in \reffig{fig:treetransformation}.
            Then, we construct a tree such that every node represents a bridgeless subswarm and every edge represents a bridge.
            The leaves of this tree represent subswarms that only contain one bridge.
            The proof of \reflem{lem:area_bridges} then mainly follows by the argument that every tree contains enough leaves for compensating
            the non-decreasing \area\ progress of the subswarms that are represented by its inner nodes.

        \subsection{Preparing for the proofs}
            
            Our arguments mainly consider only robots on the swarm's outer boundary.
            But in general, a swarm also contains inside robots.
            The proofs of \reflem{lem:area_nobridges} and \ref{lem:area_bridges} require that inside robots neither inhibit
            outer boundary robots from hopping towards the swarm's inside (\area\ decrease) nor enable additional robot hops towards the swarm's
            outside (\area\ increase).
            We formally prove this in \reflem{lem:boundarycorrect}.
            \begin{lemma}\label{lem:boundarycorrect}
                Let $\partial S$ be the outer boundary of some swarm $S$ and $S^\circ$ be its interior robots.
                We consider \DiagAB\ hops:
                \begin{enumerate}
                    \item If $r\in\partial S$ hops towards the swarm's inside, it also does this on $S$.\label{enum:bc_convexhop}
                    \item If $r\in\partial S$ does not hop towards the swarm's outside, it also does not hop on $S$.\label{enum:bc_concavehop}
                \end{enumerate}
            \end{lemma}
            Roughly speaking, in the proof we analyze how \DiagAB\ and Inhibit patterns can be generated or destroyed by adding robots
            and prove that in situations in which this would lead to a contradiction of the lemma, this cannot be done by adding robots
            only to the swarm's interior.
            For the formal proof, see \refssec{ssec:proofboundarycorrect} (appendix).

            For the \area\ progress proofs, we need a more formal description of the swarm's outer boundary.
            \paragraph{{\bfseries Quasi lines.}}
                \begin{definition}[quasi line]\label{def:quasiline}
                    We define a subboundary, called a horizontal quasi line, as follows:
                    Cf.\ \reffig{fig:quasilinedef}.
                    \begin{enumerate}
                        \item It consists only of horizontal subboundaries of length $\geq 3$ that are
                        connected by \emph{stairways} of height $1$ or $2$. (Definition of stairways, see below.)
                        \item It begins and ends with three horizontally aligned robots.
                    \end{enumerate}
                \end{definition}
                In \reffig{fig:quasilinedef}, \emph{stairways} are marked by black robots.
                They are alternating left and right turns.
                In the figure, their endpoints are marked by bicolored robots.

                Stairways that do not match the quasi line Definition~\ref{def:quasiline}, i.e., are of height $>2$, connect
                two neighboring quasi lines that can be either both horizontal, both vertical, or one horizontal and the other one vertical.
                Our algorithm performs \DiagAB\ hops at these connection points.
                In \reffig{fig:quasilinedef}, dashed lines border the according patterns
                for \emph{quasi line} 1.
                \begin{figure*}[h]
                   \begin{center}
                       \includegraphics[width=\textwidth]{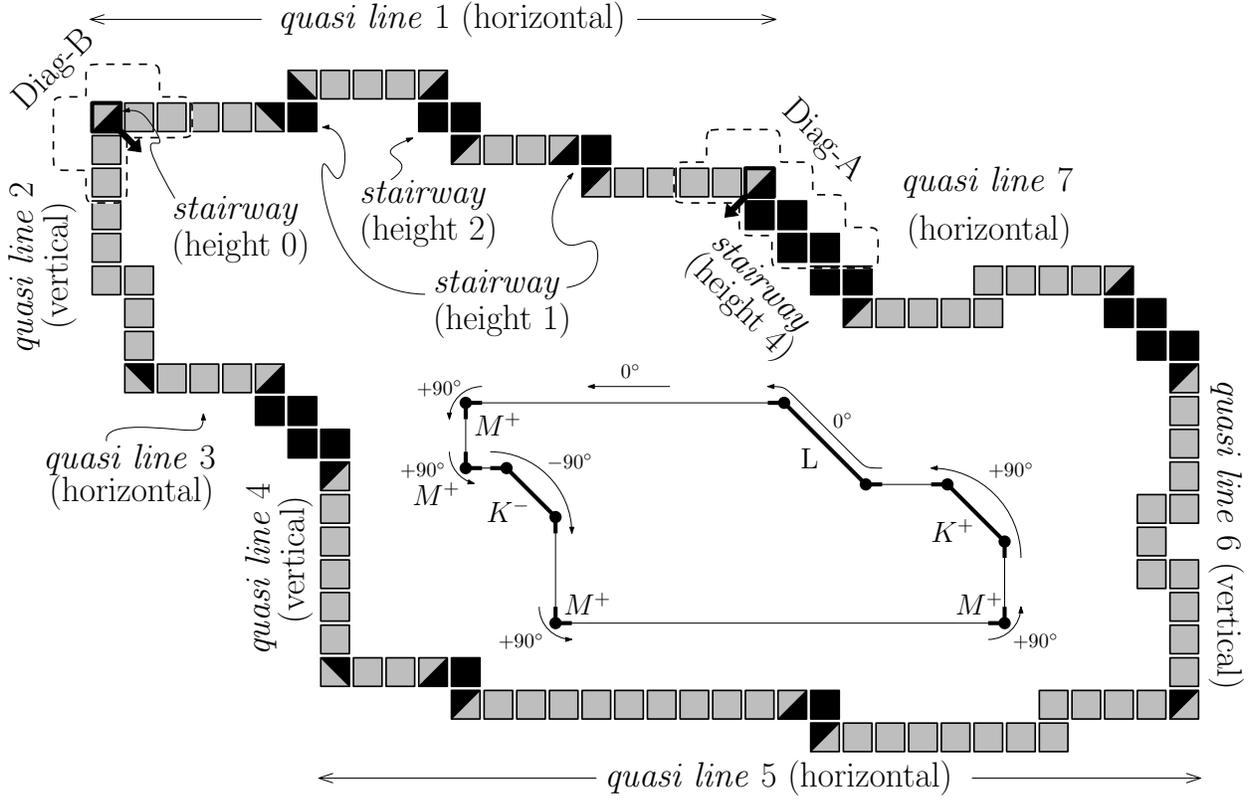}
                   \end{center}
                   \caption{Example: Quasi lines and stairways and an abstract representation of the shown boundary.}
                   \label{fig:quasilinedef}
                \end{figure*}

                In our proofs, we represent a horizontal quasi line by a horizontal line segment and a stairway
                by a diagonal one.
                For the example in \reffig{fig:quasilinedef}, the drawn polygon shows this construction for the given example swarm.
            \paragraph{{\bfseries Corners.}}
                We call the connections between different quasi lines \emph{corners}.
                These are the fat drawn parts of the polygon.
                At their endpoints, the \DiagAB\ Hop patterns match.
                The formal definition of \emph{corners} is the following (cf.\ \reffig{fig:cornerdef_ls}).
                \begin{figure}[h]
                   \begin{center}
                       \includegraphics[width=0.6\textwidth]{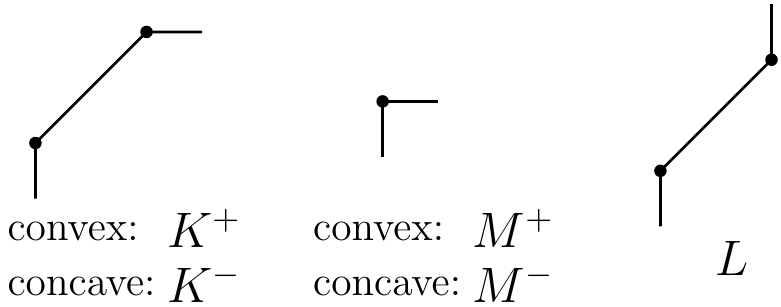}
                   \end{center}
                   \caption{Definition of \emph{corners}.}
                   \label{fig:cornerdef_ls}
                \end{figure}
                \begin{tabbing}
                    $M^+$ \=\kill
                    $K^+$ \> Corner includes a diagonal line segment and is convex (i.e., induces a $+90^\circ$ rotation).\\
                    $K^-$ \> Corner includes a diagonal line segment and is concave (i.e., induces a $-90^\circ$ rotation).\\
                    $L$ \> Corner includes a diagonal line segment and one end is convex and the other one concave (i.e.,\\\> induces $0$ rotation).\\
                    $M^+$ \> Corner does not include a diagonal line segment and is convex (i.e., induces a $+90^\circ$ rotation).\\
                    $M^-$ \> Corner does not include a diagonal line segment and is concave (i.e., induces a $-90^\circ$ rotation).
                \end{tabbing}
            \paragraph{{\bfseries Supercorners.}}
                In our strategy, though the visible pattern of some robot matches one of \DiagAB, it does not execute its hop
                if suitable inhibition patterns match.
                Remember the difference, that if for a robot a \DiagA\ pattern matches, only one matching inhibition pattern inhibits
                its hop, but for \DiagB\ two (one in each direction) inhibition patterns are required.
                \begin{figure}[h]
                   \begin{center}
                       \includegraphics[width=0.8\textwidth]{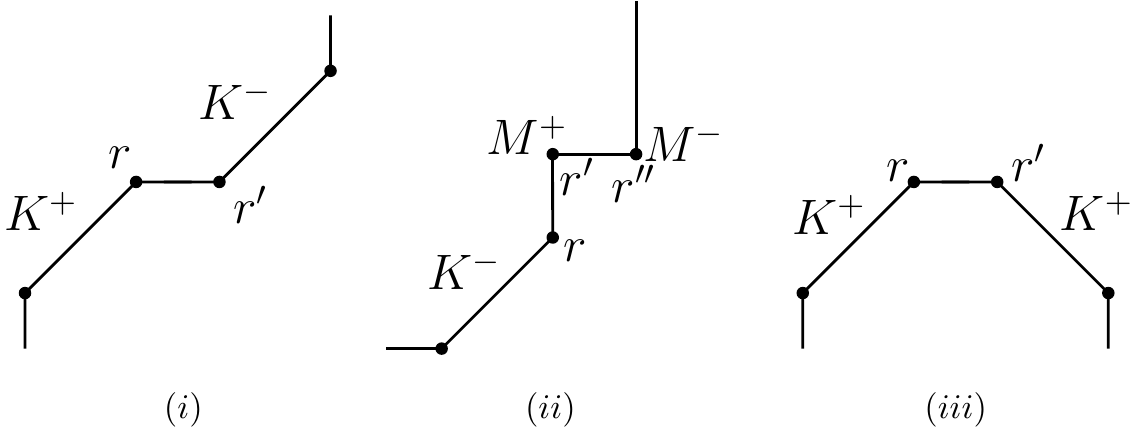}
                   \end{center}
                   \caption{Collisions in corner representation. (Significant examples)}
                   \label{fig:corner_ls_collide_ex}
                \end{figure}
                \reffig{fig:corner_ls_collide_ex} shows significant examples for this:
                $(i)$: $r,r'$ do not hop.
                $(ii)$: $r,r'$ do not hop, but the robot $r''$ does, because a colliding robot exists only in one direction on the boundary
                     (\DiagB\ pattern).
                $(iii)$: $r,'r$ execute their hops, because they both hop downwards (no collision).

                Now we construct \emph{supercorners} $\calX,\calY,\calZ$ from consecutive $K^+,$ $K^-,$ $L,$ $M^+,$ $M^-$ corners,
                in such a way that the supercorners have the property that
                only the two (respectively the one for supercorners that only consist of a single $M^\pm$ corner) endpoint
                robots actually perform their hops
                and the hops of all other included $K^+,K^-,L,M^+,M^-$ endpoint robots are inhibited by collisions with neighbors.
                Because, in order to collide, the hops must be performed in opposite directions, the consecutive corners $K,M$ along
                $\calX,\calY,\calZ$ must be alternating convex and concave.
                This means that we can define the supercorners analogously to corners (cf.~\reffig{fig:supercorner_ls_def}).
                Here, $n_{K^+},n_{M^+}$ denotes the number of $K^+$, respectively $M^+$ corners of the according supercorner and
                $n_{K^-},n_{M^-}$ denotes the same but for $K^-$ and $M^-$, respectively.
                \begin{tabbing}
                    $\calX^+$\=: \=\kill
                    $\calX^+$\>: \> $n_{K^+} + n_{M^+} = n_{K^-} + n_{M^-} +1$ \= (total $+90^\circ$ rotation)\\
                    $\calX^-$\>: \> $n_{K^+} + n_{M^+} = n_{K^-} + n_{M^-} -1$ \> (total $-90^\circ$ rotation)\\
                    $\calY$\>: \> $n_{K^+} + n_{M^+} = n_{K^-} + n_{M^-}$ \> (no total rotation)\\
                    $\calZ^+$\>: \> $n_{M^+} = 1; n_{K^+},n_{K^-},n_{M^-} =0$ \> ($+90^\circ$ rotation)\\
                    $\calZ^-$\>: \> $n_{M^-} = 1; n_{K^+},n_{K^-},n_{M^+} =0$ \> ($-90^\circ$ rotation)
                \end{tabbing}
                \begin{figure}[h]
                   \begin{center}
                       \includegraphics[width=0.45\textwidth]{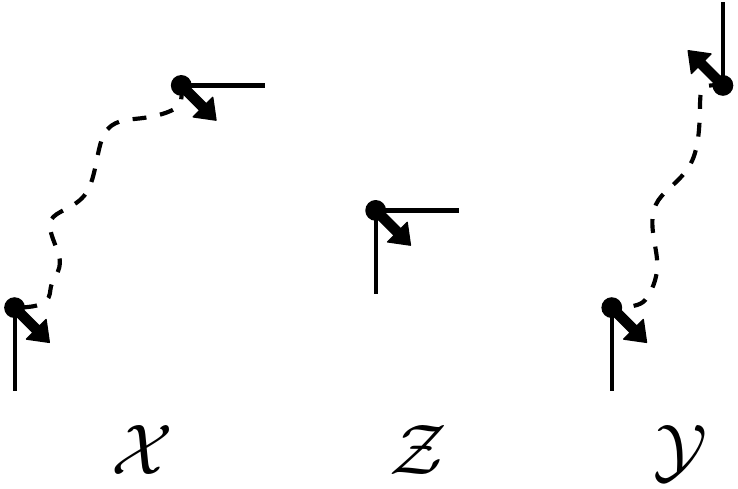}
                   \end{center}
                   \caption{Definition of \emph{supercorners} $\calX,\calY,\calZ$.}
                   \label{fig:supercorner_ls_def}
                \end{figure}
        \subsection{\emph{Bridges}}
            In \reflem{lem:area_nobridges} and \ref{lem:area_bridges}, we distinguish two kinds of swarms:
            the ones without (\reflem{lem:area_nobridges})  and the ones with \emph{bridges} (\reflem{lem:area_bridges}).
            \begin{figure}[h]
               \begin{center}
                   \includegraphics[width=0.8\textwidth]{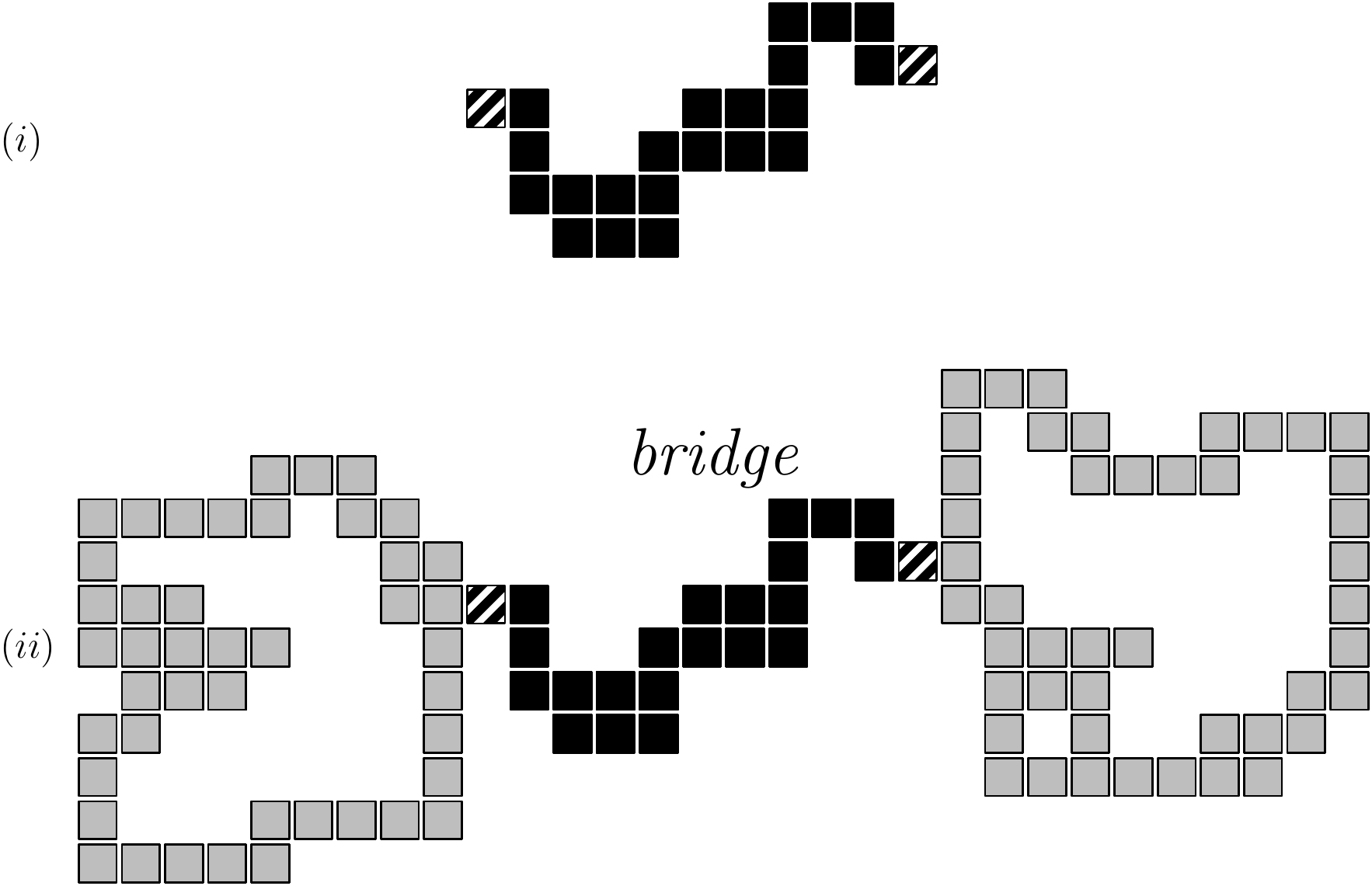}
               \end{center}
               \caption{(i): The black subswarm can be completely removed by executing a sequence of HV hops from its two endpoints (hatched). (ii): The \emph{bridge} connects two subswarms.}
               \label{fig:bridge_width1n2_ex}
            \end{figure}
            For the definition of \emph{bridges}, we start with a swarm (cf.\ \reffig{fig:bridge_width1n2_ex}.$(i)$)
            which has the property, that it contains exactly two endpoints (hatched) from which the whole swarm can be completely removed just
            by continuously executing HV hops.
            \reffig{fig:bridge_width1n2_ex}.$(ii)$: We call such a subswarm a \emph{bridge} if it connects two other subswarms
            (grey) at its endpoints.
        \subsection{\area\ progress for swarms without bridges (\reflem{lem:area_nobridges}).}
            \begin{lemma}
                \label{lem:area_nobridges}
                If a swarm $\calS$ does not contain any bridges and neither \boundary\ nor \convex\ has progress, then instead \area\ has progress.
            \end{lemma}
            \begin{proof}
                During a walk along the outer boundary of the swarm, we perform a total rotation of $360^\circ$.
                Using the above construction, it follows:
                $$ N_{K^+}+N_{M^+}=N_{K^-}+N_{M^-}+4\mathrm{,}$$
                while $N_{K^+},N_{M^+},N_{K^-},N_{M^-}$ are the total numbers of $K^+,$ $M^+,$ $K^-$ and $M^-$, respectively.
                Then, there must exist neighboring corners without collisions.
                Then, by construction, the above equation also holds for supercorners:
                $$N_{\calX^+}+N_{\calZ^+}=N_{\calX^-}+N_{\calZ^-}+4\mathrm{.}$$
                By the prerequisites of the lemma, there is no \convex\ prog\-ress.
                So, $N_{\calZ^+}\stackrel{!}{=}0$.
                Note that $N_{\calZ^-}$ can be bigger than zero if the according hop is hindered by inside robots.
                We get:
                $$N_{\calX^+}=N_{\calX^-}+N_{\calZ^-}+4\mathrm{.}$$
                By construction, every $\calX^+$ executes two hops towards the swarm's inside, while $\calX^-$ does the
                same towards the outside.
                Then, the included area behaves as follows:
                $$\Delta\mathrm{\area}=-2(N_{\calX^+}-N_{\calX^-})=-2(N_{\calZ^-}+4)\leq-8\mathrm{.}$$
                This finishes the proof.
            \end{proof}
        \subsection{\area\ progress for swarms with bridges (\reflem{lem:area_bridges}).}
        Now we generalize to a swarm that includes \emph{bridges}.
            \begin{lemma}
                \label{lem:area_bridges}
                If a swarm $\calS$ contains $k>0$ bridges and neither \boundary\ nor \convex\ has progress, then instead \area\ has progress.
            \end{lemma}
            \begin{proof}
                We classify subswarms by the number of bridges by which they are connected to the rest of the swarm.
                \paragraph{{\bfseries\#bridges $=1$ (leaves).}}
                    We first look at a subswarm $S'$ that is connected by only one bridge $B$ to the remaining part of $\calS$.
                    Later, we will call a subswarm with this property a \emph{leaf}.
                    For our analysis, we separate $S'$ from $\calS$, by splitting $B$ into two parts (e.g., this could be done by removing $\leq2$ robots).
                    Afterwards, we remove the part that is then connected to $S'$ by executing a sufficient number of HV hops.
                    The remaining subswarm $\tilde{S}'$ from $S'$ then only allows \DiagAB\ hops.
                    Now we start with the same construction as in the proof of \reflem{lem:area_nobridges}, but this time
                    applied only to the subswarm $\tilde{S}'$.
                    Accordingly, we end up in the equation:
                    $$N^{\tilde{S}'}_{\calX^+}+N^{\tilde{S}'}_{\calZ^+}=N^{\tilde{S}'}_{\calX^-}+N^{\tilde{S}'}_{\calZ^-}+4$$
                    In contrast to the proof of \reflem{lem:area_nobridges}, here $N^{\tilde{S}'}_{\calZ^+}$ can be bigger
                    than $0$, in case $B$ hinders the hops of all existing $\calZ^+$ supercorners.

                    Now we estimate the worst case for the number of diagonal hops towards the inside of $\tilde{S}'$ that could have
                    been hindered by the bridge $B$.
                    There are two cases how $B$ can hinder a diagonal hop:
                    \begin{enumerate}
                        \item It occupies white marked cells in the \DiagAB\ patterns.
                        \item It produces a collision with a robot that wants to perform a \DiagAB\ hop.
                    \end{enumerate}
                    As the width of $B$ by definition is $\leq2$, it can hinder at most two hops towards the inside of $\tilde{S}'$.
                    \begin{figure}[h]
                       \begin{center}
                           \includegraphics[width=0.65\textwidth]{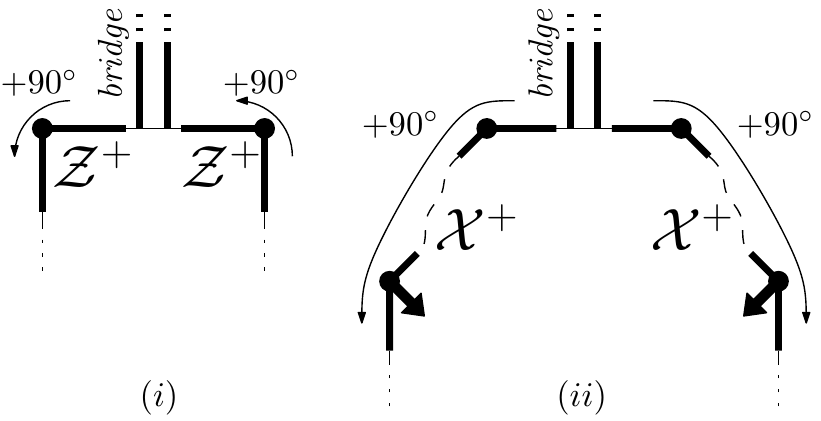}
                       \end{center}
                       \caption{A bridge hinders at most two hops towards the swarm's inside.}
                       \label{fig:hindercases}
                    \end{figure}
                    In \reffig{fig:hindercases}, we see that the $\calZ^+$ supercorners ($(i)$) provide the worst case, because here we can hinder \emph{all} hops
                    of two convex $+90^\circ$ supercorners.
                    If, e.g., hops of $\calX^+$ supercorners are hindered instead ($(ii)$),
                    then afterwards still two of their robots hop towards the swarm's inside.

                    For the worst case, we get the equation
                    \begin{eqnarray*}
                        N^{\tilde{S}'}_{\calX^+}+N^{\tilde{S}'}_{\calZ^+} & = & N^{\tilde{S}'}_{\calX^-}+N^{\tilde{S}'}_{\calZ^-}+4 \\
                        N^{\tilde{S}'}_{\calX^+}+2 & = & N^{\tilde{S}'}_{\calX^-}+N^{\tilde{S}'}_{\calZ^-}+4 \\
                        N^{\tilde{S}'}_{\calX^+}   & = & N^{\tilde{S}'}_{\calX^-}+N^{\tilde{S}'}_{\calZ^-}+2\mathrm{.}
                    \end{eqnarray*}
                    By construction, every $\calX^+$ executes two hops towards the swarm's inside, while $\calX^-$ does the
                    same towards the outside.
                    As we, because of the prerequisites of the lemma, do not have \convex\ progress, all $\calZ^-$ hops are hindered by inside robots (I.e., $N^{\tilde{S}'}_{\calZ^-}$ might be $>0$.).
                    Then the included area behave as:
                    $$\Delta\mathrm{\area}^{S'}\leq-2(N^{\tilde{S}'}_{\calX^+}-N^{\tilde{S}'}_{\calX^-})=-2(N^{\tilde{S}'}_{\calZ^-}+2)\leq-4\mathrm{.}$$
                    I.e., in $\tilde{S}'$ respectively $S'$, the included area becomes smaller by at least 4 instead 8 for swarms without bridges.
                \paragraph{{\bfseries\#Bridges $\geq1$.}}
                    We can generalize this worst case construction to a subswarm $S^\star$ that is connected by $k\geq1$ bridges to the remaining swarm.
                    For this, analogously to the one-bridge-case, we first build $\tilde{S}^\star$ from $S^\star$ by removing the bridges.

                    There, every bridge can still hinder at most 2 hops towards the subswarm's inside and, as before, the worst case
                    means, that these are $\calZ^+$ corners.
                    We get $N^{\tilde{S}^\star}_{\calZ^+}=2k$.
                    So, in total
                    \begin{eqnarray*}
                        N^{\tilde{S}^\star}_{\calX^+}+2k & = & N^{\tilde{S}^\star}_{\calX^-}+N^{\tilde{S}^\star}_{\calZ^-}+4 \\
                        N^{\tilde{S}^\star}_{\calX^+}   & = & N^{\tilde{S}^\star}_{\calX^-}+N^{\tilde{S}^\star}_{\calZ^-}+4-2k
                    \end{eqnarray*}
                    and
                    \begin{eqnarray*}
                        \Delta\mathrm{\area}^{S^\star}  & \leq & -2(N^{\tilde{S}^\star}_{\calX^+}-N^{\tilde{S}^\star}_{\calX^-})\\
                                                        & = & -2(N^{\tilde{S}^\star}_{\calZ^-}+4-2k)\\
                                                        & \stackrel{\scriptsize(\ast)}{\leq} & -8+4k\mathrm{,}
                    \end{eqnarray*}
                    where for $N^{\tilde{S}^\star}_{\calZ^-}=0$ $(\ast)$ becomes an equal sign (worst case).
                   
                    This means that in the worst case, for $2$ bridges the area does not change and for $>2$ bridges it even increases.
                    We will compensate this by showing that there are still enough subswarms that are only connected by a single bridge (leaves).

                    For showing this, we first transform the swarm to a tree.
                \paragraph{{\bfseries Tree.}}
                    \begin{figure}[h]
                       \begin{center}
                           \includegraphics[width=\textwidth]{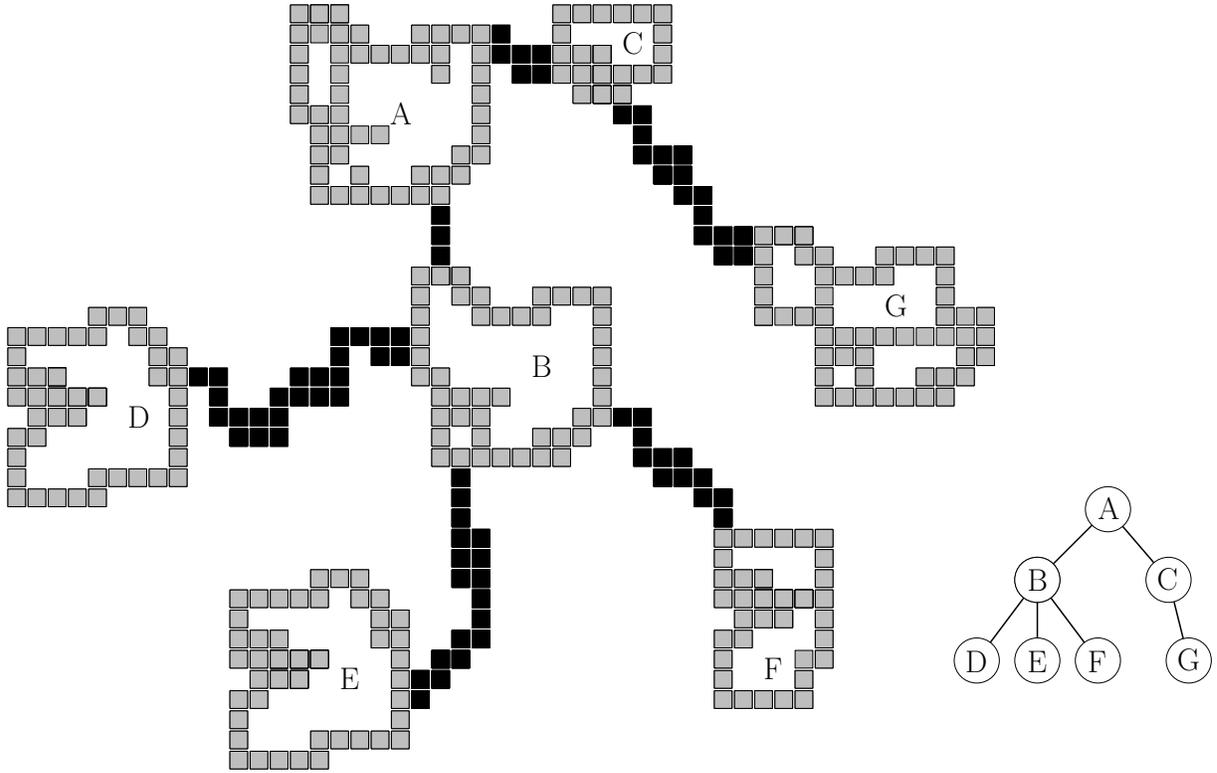}
                       \end{center}
                       \caption{Transformation of a swarm with bridges to a tree.}
                       \label{fig:treetransformation}
                    \end{figure}
                    Cf.\ \reffig{fig:treetransformation}.
                    We divide the swarm into bridges and subswarms $S_i$
                    (in the figure, the subswarms $S_i$ are named $A,B,\ldots, G$.).
                    In the graph representation, every $S_i$ becomes a node while every bridge is interpreted as an edge.
                    The resulting graph is a tree, because in our construction an existing cycle would represent an ``area'',
                    i.e., one of the subswarms $S_i$, and so
                    would contradict our construction in which every $S_i$ is a node.
                \paragraph{{\bfseries Summarized \area\ progress.}}
                    We estimate the total amount of area progress, by inductively constructing this tree:
                    We start with the root node $v_0$, which has degree $k_0$.
                    Then also $k_0$ leaves are connected to this node.
                    We know that for this $v_0$ in the worst case \area\ behaves like
                    $\Delta\mathrm{\area}^{v_0}=-8+4k_0$, which can be $\geq0$.
                    But adding the area change of the leaves, we get
                    $\Delta\mathrm{\area}\leq(-8+4k_0) -4k_0=-8$.

                    Now we inductively add the other inner nodes to the tree.
                    Then one can show that for every tree the following holds:
                    Every inner node $v$ of degree $k$ serves $k-2$ additional leaves.
                    If we associate these leaves to $v$, we get $\Delta\area^v\leq (-8+4k) - 4(k-2)=0$.
                    So $v$ does not have any bad impact on the change of total included area:
                    $\Delta\area = \Delta\area^{v_0}+\Delta\area^v\leq -8 +0 = -8$.
                    This finishes the proof.
            \end{proof}

        \reflem{lem:area_nobridges} and \ref{lem:area_bridges} immediately lead to \reflem{cor:area}.
    \subsection{Proof of \reflem{lem:boundarycorrect}}\label{ssec:proofboundarycorrect}
        \begin{proof}
            In this proof, if talking about collisions that inhibit hops, we often analyze
            the case when, in the viewing range of a robot, a \DiagA\ pattern
            changes to \DiagB\ and vice versa.
            This is relevant, because by definition a \DiagB\ hop is only inhibited of two Inhibit patterns
            match at the same time (one in each direction), while for \DiagA\ hops only one is needed.
            Now, we start with the proofs:

            \ref{enum:bc_convexhop}.)
            The hop of $r$ can be hindered by inside robots by two reasons.
            $a)$ The Hop pattern of $r$ changes:
                This is impossible, as for this outside robots must be added.
            \begin{figure}[h]
                \begin{center}
                    \includegraphics[width=0.7\textwidth]{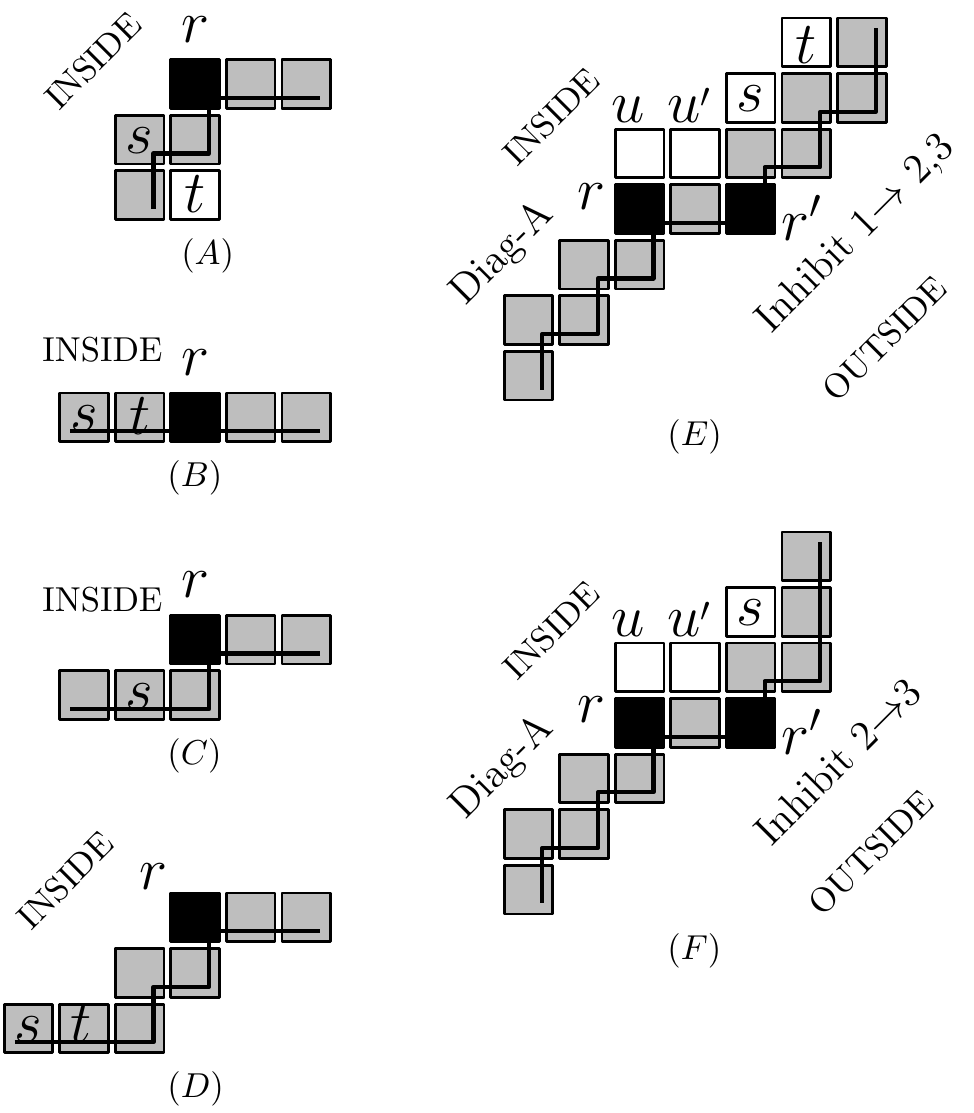}
                \end{center}
                \caption{Illustrations for the proof of \reflem{lem:boundarycorrect}.}
                \label{fig:innerbotsproof}
            \end{figure}

            $b)$ Collisions:
                $(i)$: $r$ changes from \DiagB\ to \DiagA: This is not possible because of $a)$.

                $(ii)$: Some robot $r'$ becomes part of an Inhibit pattern that makes him collide with $r$:
                    This is not possible because if $r$ checks for collisions, inhibition patterns do not include inside robots.

            \ref{enum:bc_concavehop}.) Cf.\ \reffig{fig:innerbotsproof}.
            The bold polygonal line marks the outer boundary.
            The hop of a robot $r$ that shall hop towards the swarm's outside can be enabled by inside robots by two reasons.
            $a)$ The robot positions in the viewing range of $r$ change:
                $(i)$: $r$ changes from \DiagA\ to \DiagB.
                    \reffig{fig:innerbotsproof}.$(A)$:
                    This is not possible, because for this, $s$ must be removed and $t$ occupied.
                    Both is not possible.

                $(ii)$: $r$ was located on a \emph{quasi line} and changes to \DiagAB.
                    $\alpha)$: \reffig{fig:innerbotsproof}.$(B)$: $r$ was located on a group of horizontally aligned robots.
                    This would require to remove at least the robot $t$.
                    $\beta)$: \reffig{fig:innerbotsproof}.$(C)$: $r$ was located on a stairway of height $1$.
                    This would require to remove at least the robot $s$.
                    $\gamma)$: \reffig{fig:innerbotsproof}.$(D)$: $r$ was located at an endpoint of a stairway of height $2$.
                    This would require to remove one of the robots $s,t$.
                    $\delta)$: $r$ could also not be located inside a stairway, because then the white marked cells in the \DiagAB\ patterns are not both
                    free and also cannot be emptied.

            $b)$ Collisions:
                $(i)$: $r$ was \DiagA\ and becomes \DiagB. This is not possible because of $a).(i)$. 
                $(ii)$: Previously, $r$ collided with some robot $r'$ and
                        after adding inside robots, they do not collide anymore.
                        \reffig{fig:innerbotsproof}.$(E)$: 
                        We assume, that for $r$ a \DiagA\ pattern matches (The proof for \DiagB\ is the same.).
                        Initially, the Inhibit 1 pattern inhibits the hop.
                        The cells $u,u'$ must stay empty because else the Hop pattern for $r$ would be destroyed.
                        But we can change the occupancy of the cells $s,t$:
                        $\alpha)$: If we add a robot to $t$, while $s$ stays empty, then the Inhibit 1 pattern becomes an Inhibit 2 pattern
                            and the collision is still present.
                        $\beta)$: If we add a robot to $s$, then, independent of $t$, the Inhibit 3 pattern matches
                            and the collision is still present.
                        $\gamma)$: If initially an Inhibit 2 pattern induces the collision ($(F)$), then adding a robot to $s$ makes
                        the Inhibit 2 pattern become an Inhibit 3 pattern.
                        This covers all variants, because Inhibit 3 patterns cannot be transformed.
        \end{proof}
\newpage
\begin{sloppy}
\bibliography{references}
\end{sloppy}
\end{document}